\documentclass[runningheads]{llncs}
\usepackage{graphicx,amssymb,amsmath,xspace,url}
\urldef{\mailsc}\path|eric.chacc@uai.cl,bruno.martin@unice.fr|

\newcommand{\Pol}{\textsf{P}\xspace}
\newcommand{\NP}{\textsf{NP}\xspace}
\newcommand{\NC}{\textsf{NC}\xspace}
\newcommand{\LOGDCFL}{\textsf{LOGDCFL}\xspace}

\newcommand{\AC}{\textsf{AC}\xspace}
\renewcommand{\O}[1]{\ensuremath{\mathcal{O}\left(#1\right)}\xspace}
\newcommand{\N}{\ensuremath{\mathbb{N}}\xspace}
\newcommand{\Z}{\ensuremath{\mathbb{Z}}\xspace}
\newcommand{\ie}{\emph{i.e.}\@\xspace}
\newcommand{\etal}{\emph{et\ al.}\@\xspace}
\newcommand{\myproblem}[3]{\par\noindent\textbf{Problem #1}\\[-3ex]
\begin{description}
  \item[\textsc{Instance}] #2
  \item[\textsc{Question}] #3
\end{description}
\medskip}

\begin{document}

\mainmatter

\title{Computational Complexity of Avalanches in the Kadanoff
  two-dimensional Sandpile Model}

\titlerunning{Comput. Complexity of avalanches in 2-d KSPM}

\author{Eric Goles\dag \and  Bruno Martin\ddag}

\authorrunning{E.~Goles and B.~Martin}

\institute{\dag Universidad Adolfo Iba\~nez, Av. Diagonal Las Torres 2640,
Pe\~nalolen, Santiago, Chile.\\
\ddag University of Nice-Sophia Antipolis, I3S-CNRS,  France.\\
\mailsc\\
%\url{http://www.springer.com/lncs}
}

%  \thanks{This work has been supported by the french ANR
%    programme ``\emph{Emergence dans les mod\`eles de calcul}''
%    (B.M.) and by FONDECYT 1100003, BASAL-CMM-U. de Chile,
%    Anillo ACT-88 , SFI-Santa Fe, and CNRS (E.G.).}

\toctitle{Lecture Notes in Computer Science}
\tocauthor{Authors' Instructions}
\maketitle

\begin{abstract}
  In this paper we prove that the \emph{avalanche problem} for
  Kadanoff sandpile model (KSPM) is $\Pol$-complete for
  two-dimensions.  Our proof is based on a reduction from the monotone
  circuit value problem by building logic gates and wires which work
  with configurations in KSPM.  The proof is also related to the known
  \emph{prediction problem} for sandpile which is in $\NC$ for
  one-dimensional sandpiles and is $\Pol$-complete for dimension $3$
  or greater. The computational complexity of the prediction problem
  remains open for two-dimensional sandpiles.
\end{abstract}
\section{Introduction}
Predicting the behavior of discrete dynamical systems is, in general,
both the ``most wanted'' and the hardest task.  Moreover, the
difficulty does not decrease when considering finite phase spaces.
Indeed, when the system is not solvable, numerical simulation is the
only possibility to compute future states of the system.

In this paper we consider the well-known discrete dynamical system of
sandpiles (SPM). Roughly speaking, its dynamics is as follows. Consider the
toppling of grains of sand on a (clean) flat surface, one by one. After a
while, a sandpile has formed. At this point, the simple addition of
even a single grain may cause avalanches of grains to fall down along
the sides of the sandpile. Then, the growth process of the sandpile
starts again. Remark that this process can be naturally extended to
arbitrary dimensions although for $d>3$, the physical meaning is not
clear.

The first complexity results about SPM appeared in ~\cite{gm1,gm2}
where the authors proved the computation universality of SPM. For that,
they modelled wires and logic gates with sandpiles
configurations. Inspired by these constructions, C.~Moore and
M.~Nilsson considered the \emph{prediction problem} (PRED) for SPM \ie
the problem of computing the stable configuration (fixed point)
starting from a given initial configuration of the sandpile. C.~Moore
and M.~Nilsson proved that PRED is in $\NC^3$ for dimension $1$ and
that it is \Pol-complete for $d\geq 3$ leaving $d=2$ as an open
problem \cite{moore99}. (Recall that \Pol-completeness plays for
parallel computation a role comparable to \NP-comp\-lete\-ness for
non-deterministic computation. It corresponds to problems which cannot
be solved efficiently in parallel (see~\cite{ghr95}) or, equivalently,
which are \textit{inherently sequential}). Later, P.B.~Miltersen
improved the bound for $d=1$ showing that PRED is in \LOGDCFL
($\subseteq\AC^1$) and that it is not in $\AC^{1-\epsilon}$ for any
$\epsilon>0$ \cite{miltersen07}. Therefore, in any case,
one-dimensional sandpiles are capable of (very) elementary
computations such as computing the max of $n$ bits.

Both C.~Moore and P.B.~Miltersen underline that
\begin{quote}\textit
``having a better
upper-bound than \Pol for PRED for two-dimensional sandpiles
would be most interesting.''
\end{quote}

In this paper, we address a slightly different problem: the avalanche
problem (AP). Here, we start with a monotone configuration of the
sandpile. We add a grain of sand to the initial pile. This eventually
causes an avalanche and we address the question of the complexity of
deciding whether a certain given position --initially with no grain of
sand-- will receive some grains in the future. Like for the
(PRED) problem, (AP) can be formulated in higher dimensions.
In order to get acquainted with AP, we introduce its one-dimensional
version first.
\smallskip

One-dimensional sandpiles can be conveniently represented by a finite
sequence of integers $x_1, x_2, \ldots, x_k, \ldots, x_n$. The
sandpiles are represented as a sequence of \emph{columns} and each
$x_i$ represents the number of grains contained in column $i$. In the
classical SPM, a grain falls from column $i$ to $i+1$ if and only if
the height difference $x_i-x_{i+1}\geq 2$.  Kadanoff's
sandpile model (KSPM) generalizes SPM~\cite{kadanoff89,goles08}
by adding a parameter $p$. The setting is the same except for the
local rule: one grain falls to the $p-1$ adjacent columns if the
difference between column $i$ and $i+1$ is greater than $p$.

Assume $x_k=0$, for a value of $k$ ``far away'' from the sandpile. The
avalanche problem asks whether adding a grain at column $x_1$ will
cause an avalanche such that at some point in the future $x_k\geq 1$,
that is to say that an avalanche is triggered and reaches the ``flat''
surface at the bottom.

This problem can be generalized for two-dimensional sandpiles and is
related to the question addressed by C.~Moore and P.B.~Miltersen.

In this paper we prove that in the two-dimensional case, AP is
\Pol-complete.  The proof is obtained by reduction from the Circuit
Value Problem where the circuit only contains monotone gates --- that
is, AND's and OR's (see section~\ref{sec:2d} for details).
\smallskip

We stress that our proof for the two-dimensional case needs some
further hypothesis/constraints for monotonicity and determinism (see
section~\ref{sec:2d}). If both properties are technical requirements
for the proof's sake, monotonicity also has a physical
justification. Indeed, if KSPM is used for modelling real physical
sandpiles, then the image of a monotone non-increasing configuration
has to be monotone non-increasing since gravity is the only force
considered here. We have chosen to design the Kadanoff automaton for
$d=2$ by considering a certain definition of the three-dimensional
sandpile which does not correspond to the one of Bak's \emph{et al.}
in~\cite{bak88}. This hypothesis is not restrictive. It is just used
for constructing the transition rules. Bak's construction was done
similarly. Nevertheless, our result depends on the way the three
dimensional sandpile is modelled. In our case, we have decided to
formalise the sandpile as a monotone decreasing pile in three
dimensions where $x_{i,j}\geq\max\{x_{i+1,j},x_{i,j+1}\}$ (here
$x_{i,j}$ denotes the sand grains initial distribution) together with
Kadanoff's avalanche dynamics ruled by parameter $p$. The pile $(i,j)$
can give a grain either to every pile $(i+1,j),\ldots,(i+p-1,j)$ or to
every pile $(i,j+1),\ldots,(i,j+p-1)$ if the monotonicity is not
violated. With such a rule and if we use the height difference for
defining the monotonicity, we can define the transition rules of the
automaton for every value of the parameter $p$.

In the case where the value of the parameter $p$ equals $2$, we find
in our definition of monotonicity something similar with Bak's SPM in
two dimensions. Actually, both models are different because the
definitions of the three dimensional piles differ. That is the reason
why we succeed in proving the \Pol-completeness result which remains
an open problem with Bak's definition.
\smallskip

The paper is organized as follows. Section~\ref{sec:defs} introduces
the definitions of the Kadanoff sandpile model in one dimension and
presents the avalanche problem. Section~\ref{sec:2d} generalizes the
Kadanoff sandpile model in two dimensions and presents the avalanche
problem in two dimension, which is proved \Pol-complete for any value
of the Kadanoff parameter $p$. Finally, section~\ref{sec:ccl}
concludes the paper and proposes further research directions.
\section{Sandpiles and Kadanoff model in one dimension}\label{sec:defs}
A sandpile \emph{configuration} is a distribution of sand grains over
a lattice (here \Z). Each site of the lattice is associated with an
integer which represents its sand content. A configuration is
\emph{finite} if only a finite number of sites has non-zero sand
content.  Therefore, in the sequel, a finite configuration on \Z will
be identified with an ordered sequence of integers $x_1,x_2, \ldots,
x_n$ in which $x_1$ (resp. $x_n$) is the first (resp. the last) site
with non-zero sand content. A configuration $x$ is \emph{monotone} if
$\forall i\in\Z$, $x_i\geq x_{i+1}$. A configuration $x$ is
\emph{stable} if $\forall i\in\Z$, $x_i-x_{i+1}< p$ \ie if the
difference between any two adjacent sites is less than Kadanoff's
parameter $p$. Let SM$(n)$ denote the set of stable monotone
configurations of the form $^\omega x_1,x_2,\ldots, x_{n-1},
x_n^\omega$ and of length $n$, for $x_i\in\N$.

Given a configuration $x$, $a\in\N$ and $j\in\Z$, we use the notation
$^\omega ax_j$ (resp. $x_ja^\omega$)
to say that $\forall i\in\Z$, $i<j\rightarrow x_i=a$
(resp. $\forall i\in\Z$, $i>j\rightarrow x_i=a$).

Finally, remark that any configuration $^\omega x_1,x_2,\ldots,
x_{n-1}, x_n0^\omega$ can be identified with its \emph{height
  difference} sequence $^\omega 0,(x_1-x_2), \ldots, (x_{n-1}-x_n), x_n,
0^\omega$ \enspace.  \smallskip

Consider a stable monotone configuration $^\omega x_1,x_2, \dots,
x_n^\omega$. Adding one more sand grain, say at site $i$, may cause
that the site $i$ topples some grains to its adjacent sites. In their
turn the adjacent sites receive a new grain of sand and may also topple,
and so on. This phenomenon is called an \emph{avalanche}. The
avalanche ends when the system evolves to a new stable configuration.

\begin{figure}
  \centering
  \includegraphics[scale=.6]{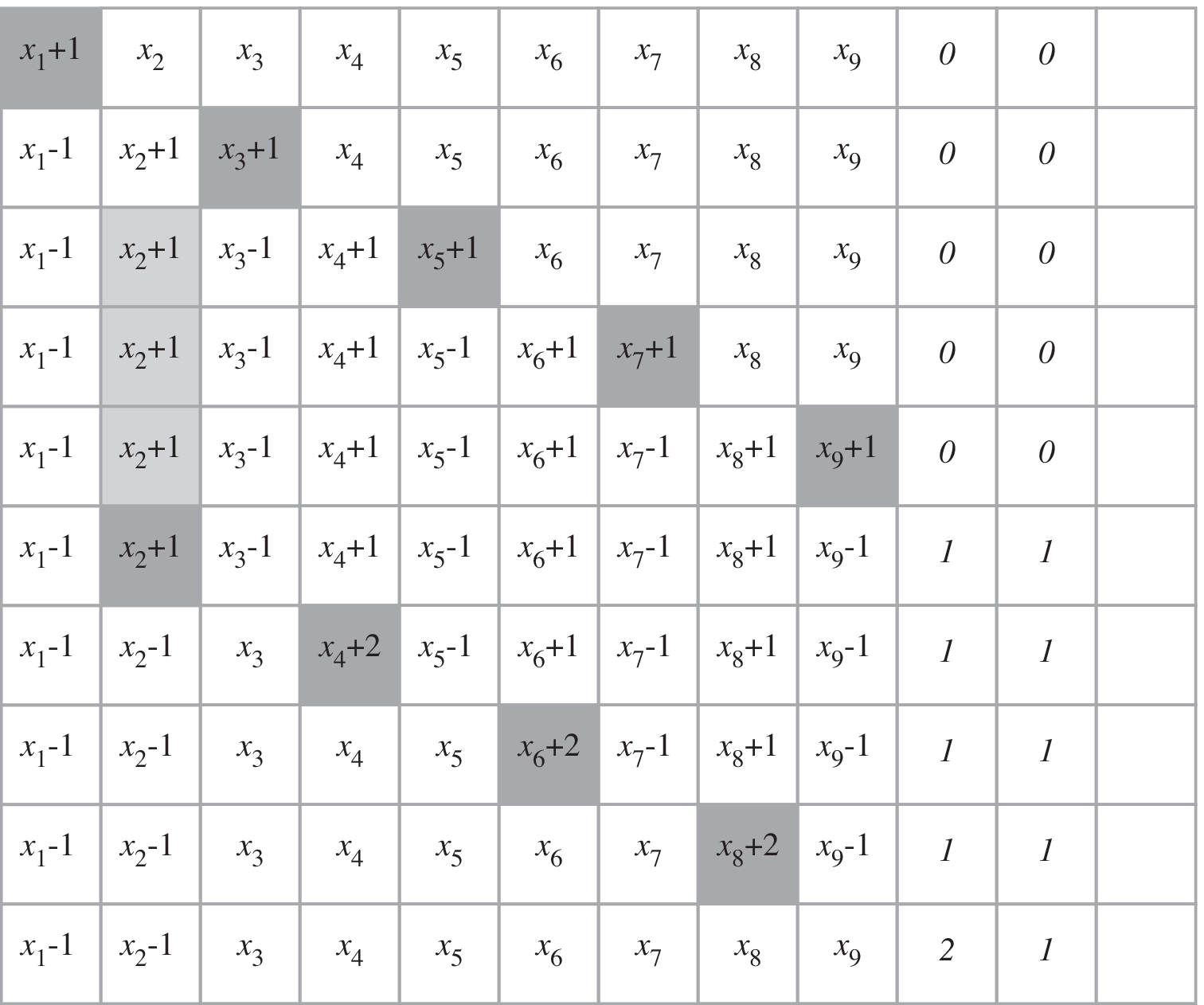}
  \caption{Avalanches for $p=3$ with 9 columns. Here, $x_i+1$
    (resp. $x_i+2$) indicates that column $i$ has received some grains
    once (resp. twice), $x_i-1$ that column $i$ has given some grains
    according to the dynamics; a dark shaded site indicates the
    toppling site, a light shaded site indicates a site that could
    topple in the future. Times goes top-down.}
  \label{fig:avalanche}
\end{figure}

In this paper, topplings are controlled by the \emph{Kadanoff's
  parameter} $p\in\N$ which completely determines the model and its
dynamics. In KSPM$(p)$, $p-1$ grains will fall from site $i$ if
$x_i-x_{i+1}\geq p$ and the new configuration becomes
\[
^\omega x_1\cdots
(x_{i-1})(x_i-p+1)(x_{i+1}+1)\cdots(x_{i+p-2}+1)(x_{i+p-1}+1)(x_{i+p})\cdots
x_n0^\omega
\enspace.
\]
In other words, the site $i$ distributes one grain to each of its
$(p-1)$ right adjacent sites. Equivalently, if we mesure the height
differences after applying the dynamics, we get
$(h_{i-1}+p-1)(h_i-p)(h_{i+1})(h_{i+2})\cdots
(h_{i+p-2})(h_{i+p-1}+1)\enspace,$ where $h_{i-1}=(x_{i-1}-x_i)$ and
all remaining heights do not change. In other words, the height
difference $h_i$ gives rise to an increase of $(p-1)$ grains of sand
to height $h_{i-1}$ and an increase of one grain to height
$h_{i+p-1}$.

We consider the problem of deciding whether some column on the right
of column $x_{n}$ (more precisely for column $x_k$ for $n<k\leq
n+p-1$) will receive some grains according to the Kadanoff's
dynamics. Since the initial configuration is stable, it is not difficult
to prove that avalanches will reach at most the column $n+p-1$
(see figure~\ref{fig:avalanche} for example).

%%%(POUR LE PAPIER ON POURRAIT LE DEMONTRER PLUS FORMALLEMENT :::)

Remark that given a configuration, several sites could topple at the
same time. Therefore, at each time step, one might have to decide
which site or which sites are allowed to topple. According to the
update policy chosen, there might be different images of the same
configuration. However, it is known ~\cite{goles02} that for any given
initial number of sand grains $n$, the orbit graph is a lattice and
hence, for our purposes, we may only consider one decision problem to
formalize AP:

\myproblem{AP}{A configuration $x\in$ SM$(n)$ and $k\in\N$
  s.t. $n\!<k\!\leq n+p-1$}{Does there exist an avalanche such that
  $x_k\geq 1$?}

Let us consider some examples.  Let $p=3$ and consider a stable
bi-infinite configuration such that its height differences is as
follows $^\omega00\underline{2}2022120000^\omega$. We add just one
grain at $x_1$ (the site underlined in the configuration).  Then, the
next step is $^\omega02\underline{0}21222120000^\omega$. And so in one
step we see that no avalanche can be triggered, hence the answer to
$AP$ is negative.
As a second example, consider the following sequence of height
differences (always with $p=3$):
$^\omega0\underline{3}122122221201200^\omega$. There are several possibilities for
avalanches from the left to the right but none of them arrives to the
0's region. So the answer to the decision problem is still negative. To get
an idea of what happens for a positive instance of the problem,
consider the following initial configuration: $^\omega0
\underline{3}12222100^\omega$ with parameter $p=3$.
\smallskip

The full proof of Theorem~\ref{th:1d} is a bit technical and will
be given in the journal version of the paper.

\begin{theorem}\label{th:1d}
AP is in $\NC^1$ for KSPM in dimension $1$ and $p>1$.
\end{theorem}
\begin{proof}[Sketch of the proof.]
  The first step is to prove that, in this situation, the Kadanoff's
  rule can be applied only once at each site for any initial monotone
  stable configuration.  Using this result one can see that a site $k$
  such that $x_k=0$ in the initial configuration and $x_k>1$ in the
  final one, must have received grains from site $k-p$. This site, in
  its turn, must have received grains from $k-2p$ and so on until a
  ``firing'' site $i$ with $i\in[\![1,p-1]\!]$. The height difference for
  all of these sites must be $p-1$.  The existence of this sequence
  and the values of the height differences can be checked by a
  parallel iterative algorithm on a PRAM in time \O{\log n}.
\end{proof}

\section{Sandpiles and Kadanoff model in two dimensions}
\label{sec:2d}
There are several possibilities to define extensions of the Kadanoff
dynamics to two dimensional sandpiles. Let us first extend the basic
definitions introduced in section~\ref{sec:defs}.
\smallskip

A two-dimensional sandpile \emph{configuration} is a distribution of
grains of sand over the $\mathbb{N}\times\mathbb{N}$ lattice.  As in
the one-dimensional case, a configuration is \emph{finite} if only a
finite number of sites has non-zero sand content. Therefore, in the
sequel, a finite configuration on $\mathbb{N}\times\mathbb{N}$ will be
identified by a mapping from $\mathbb{N}\times\mathbb{N}$ into
$\mathbb{N}$, giving a number of grains of sand to every position in
the lattice. Thus, a configuration will be denoted by $x_{i,j}$ as
$(i,j)\mapsto\mathbb{N}$. A configuration $x$ is \emph{monotone} if
$\forall i,j\in\mathbb{N}\times\mathbb{N}$, $x_{i,j}$ is such that
$x_{i,j}\geq 0$ and $x_{i,j}\geq\max\{x_{i+1,j},x_{i,j+1}\}$. So we
have a monotone sandpile, in the same sense as in~\cite{duchi06}. A
configuration $x$ is \emph{horizontally stable}
(resp. \emph{vertically}) if $\forall
i,j\in\mathbb{N}\times\mathbb{N},\quad x_{i,j}-x_{i+1,j}<p$
(resp. $\forall i,j\in\mathbb{N}\times\mathbb{N},\quad
x_{i,j}-x_{i,j+1}<p$) and is \emph{stable} if it is both horizontally
and vertically stable. In other words, it is a generalisation of the
Kadanoff model in one dimension, that is the configuration is stable
if the difference between any two adjacent sites is less than the
Kadanoff parameter $p$. To this configuration, we apply the Kadanoff
dynamics for a given integer $p\geq 1$. The application can be done if
and only if the new configuration remains monotone.
Example~\ref{ex:kspm2d} illustrates the case which violates the
condition of monotonicity of the Kadanoff dynamics.
\begin{example}\label{ex:kspm2d}
Consider the initial configuration given in the bottom left matrix of
the following figure
  {\tiny\[\begin{array}{ccc}
    \begin{matrix}
      0&1&0&0\\
      \boxed{2}&3&0&0\\
      8&4&2&2\\
      8&4&3&2
    \end{matrix}&&\\
    &&\\
    \uparrow v&&\\
    &&\\
    \begin{matrix}
      0&0&0&0\\
      2&2&0&0\\
      8&\boxed{6}&2&2\\
      8&4&3&2
    \end{matrix}
    &\stackrel{h}{\rightarrow}&
    \begin{matrix}
      0&0&0&0\\
      2&2&0&0\\
      8&4&3&\boxed{3}\\
      8&4&3&2
    \end{matrix}\\
  \end{array}\]}
Values count for the number of grains of a site.  We see
that we cannot apply the Kadanoff's dynamics for a value of
parameter $p=3$ from the boxed site. Indeed, the resulting
configurations do not remain monotone neither by applying the
dynamics horizontally nor vertically (resp. $\uparrow\! v$ and
$\stackrel{h}{\rightarrow}$). A site which violates the condition
has been boxed in the resulting configurations (it might be not
unique).\qed
\end{example}
Recall that the Kadanoff operator applied to site $(i,j)$ for a
given $p$ consists in giving a grain of sand to any site in the
horizontal or vertical line, i.e $\{(i,j+1), ....(i,j+p-1)\}$ or
$\{(i+1,j), ...(i+p-1,j)\}$.

Similarly to the one-dimensional case, we associate to the previous
avalanches their height difference. Any configuration can be
identified by the mapping of its \emph{horizontal height difference}
(resp. vertical): $h_{\rightarrow}:(i,j)\mapsto x_{i,j}-x_{i+1,j}$
(resp. $h_{\uparrow}:(i,j)\mapsto x_{i,j}-x_{i,j+1}$). The height
difference allows to define the notions of monotonicity and stability
in a straightforward way. However, notice that when considering the
dynamics defined over height differences, we work with a different
lattice though isomorphic to the initial one. The relationship between
them is depicted on figure~\ref{fig:chenilles}.

For a better understanding of the dynamics,
recall that in one dimension an avalanche at site $i$ changes the
heights of sites $i-1$ and $i+p-1$. In two dimensions, there are
height changes on the line but also to both sides of it. The dynamics
is simpler to depict than to write it down formally. It will be presented
throughout examples and figures in the sequel.
An example of the Kadanoff's dynamics applied horizontally
(resp. vertically) is given in figure~\ref{fig:chenillesHVp=4}.  More
precisely, the Kadanoff's dynamics for a value of parameter $p=4$ is
depicted in figure~\ref{fig:Chenilles-p=4}. Observe that we do not
need to take into account the number of grains of sand in the
columns. It sufficies to take the graph of the edges adjacent to each
site (depicted by thick lines) and to store the height differences.
So, from now on, we will restrict ourselves to the lattice and to the
dynamics defined on the height differences.
In figure~\ref{fig:Chenilles-p=4}, we only keep the information
required for applying the dynamics in the simplified view.  In fact,
the local function is depicted by figure~\ref{fig:chenilles} that we
will call \emph{Chenilles} (horizontal and vertical, respectively).

Figure~\ref{fig:chenilles} explains how the dark site with coordinates
$(i,j)$ with a height difference of $p$ gives grains either
horizontally (figure~\ref{fig:chenilles} left) or vertically
(figure~\ref{fig:chenilles} right).
% Observe that we also get a kind of number conserving automaton.
%
\begin{example}[Obtaining Bak's]
  In the case $p=2$ and if we assume the real sandpile is defined as
  in~\cite{duchi06}
  (i.e. $x_{i,j}\geq\max\{x_{i+1,j},x_{i,j+1}\}$), we get the templates
  from figure~\ref{fig:chenilleBak}.\qed
\end{example}

In order to be applied, the automaton's dynamics requires to test if
the local application gives us a non-negative configuration.
\subsection{\Pol-completeness}
Changing from dimension $1$ to $2$ (or greater), the statement of AP
has to be adapted. Consider a finite configuration $x$ which is
non-zero for sites $(i,j)$ with $i,j\geq0$, stable and monotone and
let $Q$ be the sum of the height differences. Let us denote by $n$ the
maximum index of non-zero height differences along both axis. Then,
SM$(n)$ denotes the set of monotone stable configurations of the form
given by a lower-triangular matrix of size $n\times n$. To generalise
the avalanche problem in two dimensions, we have to find a generic
position which is far enough from the initial sandpile but close
enough to be attained. To get rough bounds, we have followed the
following approach. For the upper bound, the worst case occurs when
all the grains are arranged on a single site (with a height difference
of $Q$) which is at an end of one of the axis and they fall down. For
the lower bound, it is the same reasoning, except that the pile
containing the grains is at the origin. Thus, we may restate our
decision problem as follows:

\myproblem{AP (dimension 2)}{A configuration $x\in$ SM$(n)$,
  $(k,\ell)\in\N\times\N$ such that $x_{k,\ell}=0$ and
%  $\|(k,\ell)\|\leq\sqrt{2}Q$}{Does there exist a trajectory of the
  $\frac{\sqrt{2}}{2}n\leq\|(k,\ell)\|\leq n+Q$ (where $Q$ is the sum of the height
  differences).}{Does there exist an avalanche (obtained by using the
  vertical and horizontal chenilles) such that $x_{k,\ell}\geq 1$?}
where $\|.\|$ denotes the standard Euclidean norm.  \smallskip

% Clearly the bound $\sqrt{2}Q$ has been chosen large enough to
% ensure that outside this diameter every position remains at zero value
% (no grain will arrive).

To prove the \Pol-completeness of AP we will proceed by
reduction from the monotone circuit value problem (MCVP), i.e given a
circuit with $n$ inputs $\{\alpha_1, ...,\alpha_n\}$ and logic gates
AND, OR we want to answer if the output value is one or zero (refer
to~\cite{ghr95} for a detailed statement of the problem). NOT gates
are not allowed but the problem remains \Pol-complete for the
following reason: using De Morgan's laws $\overline{a\wedge
  b}=\overline{a}\vee\overline{b}$ and $\overline{a\vee
  b}=\overline{a}\wedge\overline{b}$, one can shift negation back
through the gates until they only affect the inputs themselves. For
the reduction, we have to construct, by using sandpile configurations,
wires (figure~\ref{fig:wire}), logic AND gates
(figure~\ref{fig:andgate}), logic OR gates (figure~\ref{fig:orgate}),
cross-overs (figure~\ref{fig:crossover}) and signal multipliers for
starting the process (figure~\ref{fig:sigmul}). We also need to define
a way to deterministically update the network; to do this, we can
apply the chenille's templates any way such that it is spatially
periodical, for instance from the left to the right and from the top
to the bottom. Our main result is thus:

\begin{theorem}
  $AP$ is \Pol-complete for KSPM in dimension two and any $p\geq 2$.
\end{theorem}
\begin{proof}
The fact that our problem is in \Pol is already known since C.~Moore
and M.~Nilsson paper~\cite{moore99}. The proof is done by proving that
the total number of avalanches required to relax a sandpile is
polynomial in the system size. The remaining open problem in their
study was the case $d=2$ for which they wrote ``\emph{The reader may
  [...] find a clever embedding of non-planar Boolean circuits}'',
which is precisely what will be done hereafter.  For the
reduction, one has to take an arbitrary instance of (MCVP) and to
build an initial configuration of a sandpile for the Kadanoff's
dynamics for $p=2$ (or greater). Remark that, in the case $p=2$, KSPM corresponds to Bak's model~\cite{bak88} in two dimensions
with a sandpile such that $x_{i,j}\geq\max\{x_{i+1,j},x_{i,j+1}\}$.
To complete the proof, we have to design:
\begin{itemize}
\item a wire (figure~\ref{fig:wire});
\item the crossing of information (figure~\ref{fig:crossover});
\item a AND gate (figure~\ref{fig:andgate});
\item a OR gate (figure~\ref{fig:orgate});
\item a signal multiplier (figure~\ref{fig:sigmul}).
\end{itemize}
The construction is shown graphically for $p=2$ but can be done for
greater values. For $p=2$, the horizontal and vertical chenilles are
given in figure~\ref{fig:chenilleBak}. According
to~\cite{goldschlager77}, the reduction is in \textsf{NC} since MCVP
is logspace complete for \Pol.  Recall that the decision problem only
adds a sand grain to one site, say $(0,0)$. To construct the
entry vector to an arbitrary circuit we have to construct from the
starting site wires to simulate any variable $\alpha_i=1$. (If
$\alpha_i=0$ nothing is done: we do not construct a wire from the
initial site. Else, there will be a wire to simulate the value 1).
\end{proof}
\begin{remark}
  For $p\geq 3$, the construction of the AND gate is easier than for
  $p=2$. The dynamics is obtained from figure~\ref{fig:chenilles} and
  the construction of an AND gate is depicted on
  figure~\ref{fig:and-p=3}.
\end{remark}

\section{Conclusion and future work}\label{sec:ccl}

We have proved that the avalanche problem for the KSPM model in two
dimensions is \Pol-complete with a sandpile defined as
in~\cite{duchi06} and for every value of the parameter $p$. Let us
also point out that in the case where $p=2$, this model corresponds to
the two dimensional Bak's model with a pile such that $x_{i,j}\geq 0$
and $x_{i,j}\geq\max\{x_{i+1,j},x_{i,j+1}\}$. In this context, we also
proved that this physical version (with a two dimension sandpile
interpretation) is \Pol-complete. It is important to notice that, by
directly taking the two dimensional Bak's tokens game (given a graph
such that a vertex has a number of token greater or equal than its
degree, it gives one token to each of its neighbors), its computation
universality was proved in~\cite{gm2} by designing logical gates in
non-planar graphs. Furthermore, by using the previous construction,
C. Moore \etal proved the \Pol-completeness of this problem for
lattices of dimensions $d$ with $d\geq 3$. But the problem remained
open for two dimensional lattices. Furthermore, it was proved
in~\cite{goles06} that, in the above situation, it is not possible to
build circuits because the information is impossible to cross. The two
dimensional Bak's operator corresponds, in our framework, to the
application of the four rotations of the template (see
figure~\ref{fig:BakRotate}). But this model is not anymore the
representation of a two dimensional sandpile as presented
in~\cite{duchi06}, that is with $x_{i,j}\geq 0$ and
$x_{i,j}\geq\max\{x_{i+1,j},x_{i,j+1}\}$.

To define a reasonable two dimensional model, consider a monotone
sandpile decreasing for $i\geq0$ and $j\geq0$. Over this pile we
define the extended Kadanoff's model as a local avalanche in the
growing direction of the $i-j$ axis such that monotonicity is
allowed. Certainly, one may define other local applications of
Kadanoff's rule which also match with the physical sense of
monotonicity. For instance, by considering the set
${(i+1,j),(i+1,j+1),(i,j+1)}$ as the sites to be able to receive
grains from site $(i,j)$.  In this sense it is interesting to remark
that the two dimensional sandpile defined by Bak (i.e for nearest
neighbors, also called the von Neumann neihborhood, a site gives a
token to each of its four neighbors if and only if it has enough
tokens) can be seen as the application of the Kadanoff rule for $p=2$
by applying to a site, if there are at least four tokens, the
horizontal $(\rightarrow)$ and the vertical $(\downarrow)$ chenille
simultaneoulsly (see figure~\ref{fig:BakRotate}).  Similarily, for an
arbitrary $p$, one may simultaneously apply other conbinations of
chenilles which, in general, allows us to get \Pol-complete
problems. For instance, when there are enough tokens, the applications
of the four chenilles (i.e. $\leftarrow$,$\rightarrow$,$\uparrow$ and
$\downarrow$) gives raise to a new family of local templates called
\emph{butterflies} (because of their four wings). It is not so
difficult to construct wires and circuits for butterflies. Hence, for this
model of sandpiles, the decison problem will remain \Pol-complete.
One thing to analyze from an algebraic and complexity point of view is
to classify every local rule derivated from the chenille
application. Further, one may define a more general sandpile dynamics
which contains both Bak's and Kadanoff's ones: i.e given an integer $p
\geq2$, we allow the application of every Kadanoff's update for $q\leq
p$. We are studying this dynamics and, as a first result, we observe
yet that in one dimension there are several fixed points and also,
given a monotone circuit with depth $m$ and with $n$ gates, we may
simulate it on a line with this generalized rule for a given $p\geq
m+n$.

For the one-dimensional avalanche problem as defined in
section~\ref{sec:defs}, it can be proved that it belongs tho the class
\NC for $p=2$ and that it remains in the same class when the first $p$
columns contain more than one grain (\ie that there is no hole in the
pile). We are in the way to prove the same in the general case.

% AJOUT D'ERIC mais je ne sais pas s'il faut le mettre ou le garder
% pour plus tard.
% Another generalization : we believe that for almost any reasonable
% definition of two-dimensional sandpile like Bak's or Phan's and for
% dynamics which keep the monotonicity of the pile, the corresponding
% problems may be \Pol-complete.

\section*{Acknowledgements}

We thank Pr.~Enrico Formenti for helpful discussions and comments
while Pr.~Eric Goles was visiting Nice and writing this paper.

\bibliographystyle{plain}
\bibliography{KSPM.bib}
\newpage
\begin{figure}
  \centering
  \includegraphics[scale=.45]{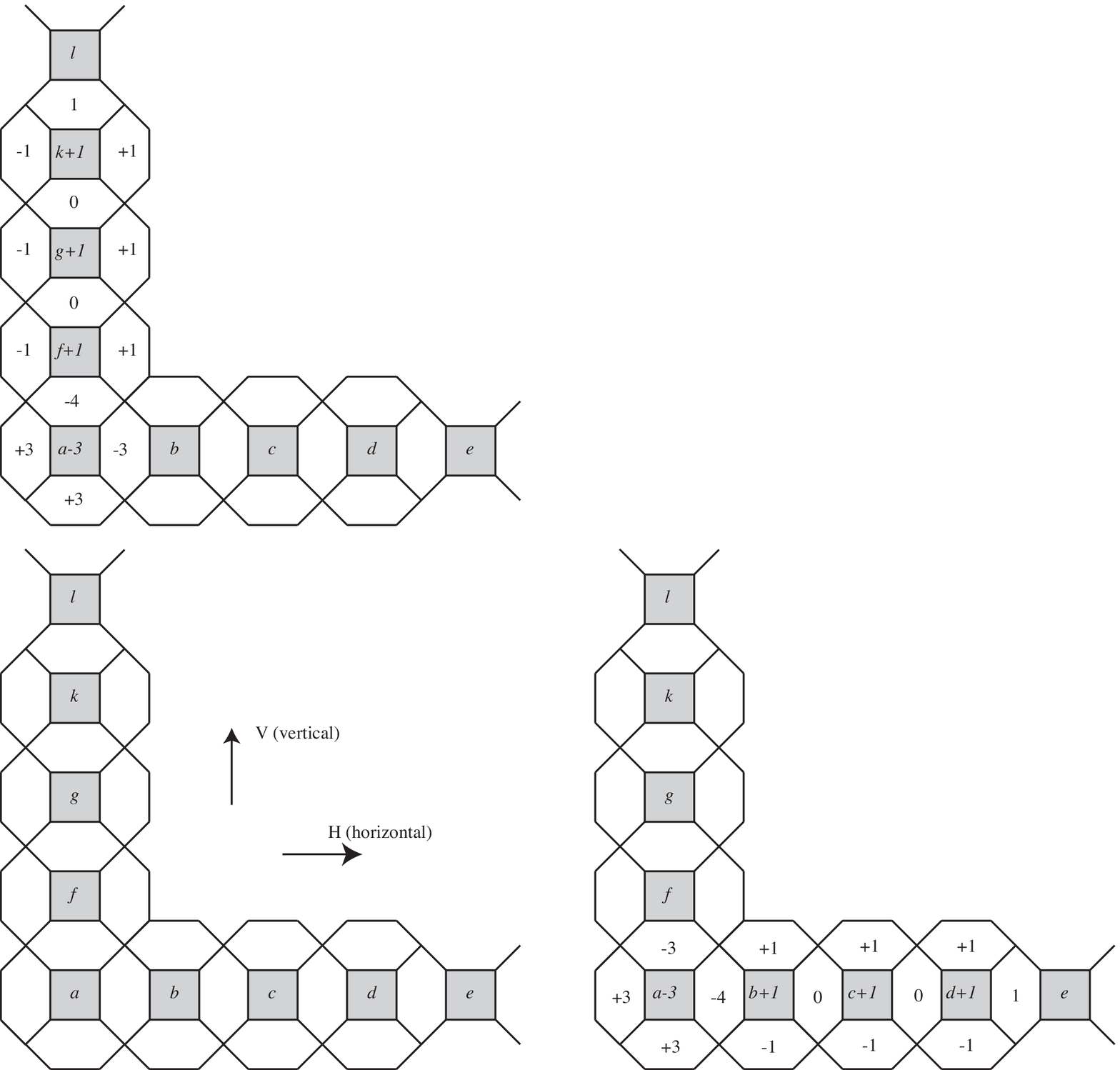}
  \caption{Horizontal and vertical chenilles for $p=4$. Shaded squares
    count the number of grains on each column and the hexagons between
    the squares the height difference between the corresponding two
    adjacent columns. The initial configuration is on the
    bottom-left. The Kadanoff's dynamics is applied from the shaded
    site labelled $a$ horizontally or vertically (resp. $\uparrow\! V$
    and $\stackrel{H}{\rightarrow}$) to get the resulting
    configurations.}
  \label{fig:chenillesHVp=4}
\end{figure}

\begin{figure}
  \centering
  \includegraphics[scale=.5]{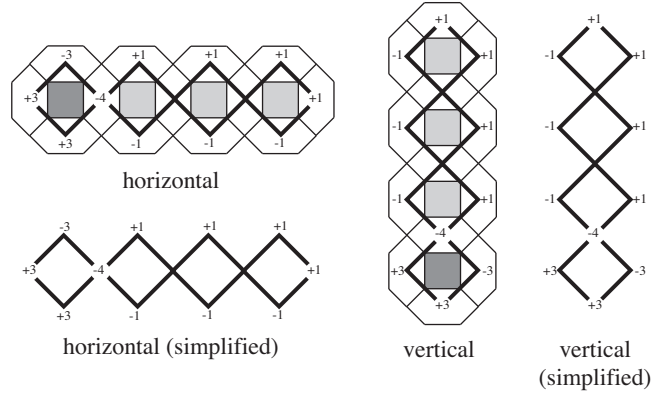}
  \caption{Horizontal and vertical chenilles for $p=4$.  The dynamics
    is applied to the dark-shaded site (the leftmost one on the left
    part of the figure and the lowest one on the right part of the
    figure). The numbers express the height differences after the
    application of the dynamics. The two simplified views remove the
    number of sand grains information and only keeps the height
    difference information. It corresponds to a change in the lattice
    structure if the grains are considered or the height difference.}
    \label{fig:Chenilles-p=4}
\end{figure}
\begin{figure}
 \centering
\includegraphics[scale=.55]{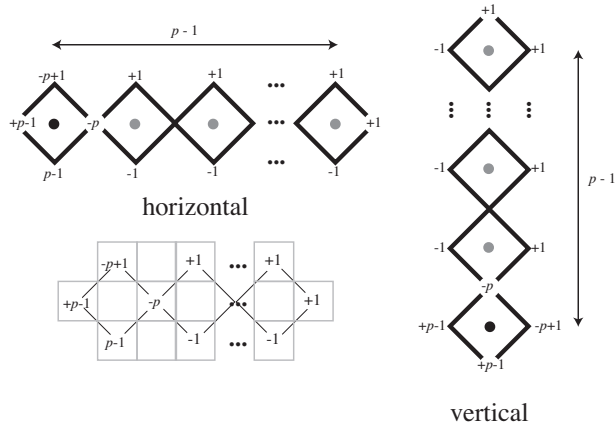}
\caption{Horizontal resp. vertical (on the left resp. on the right)
  chenilles in the $\N\times\N$ lattice for arbitrary parameter $p\geq
  2$. The site $(i,j)$ --denoted by a black bullet-- gives one grain
  of sand to the site $(i+p-1,j)$ horizontally resp. one grain of sand
  to the site $(i,j+p-1)$ vertically. The figure gives the height
  differences of this dynamics and the change of the lattice structure
  between the dynamics on the grains and the corresponding dynamics on
  the height difference. In the sequel, we will adopt the
  representation on the left and on the bottom for defining the templates.}
  \label{fig:chenilles}
\end{figure}

  \begin{figure}
    \centering
    \includegraphics[scale=.65]{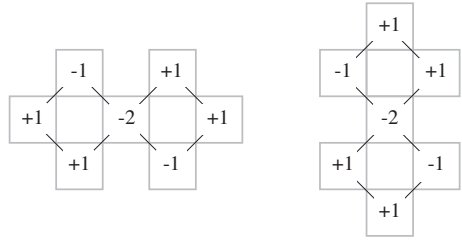}
    \caption{Templates for Bak's dynamics with $p=2$.}
    \label{fig:chenilleBak}
  \end{figure}

  \begin{figure}
    \centering
    \includegraphics[scale=.5]{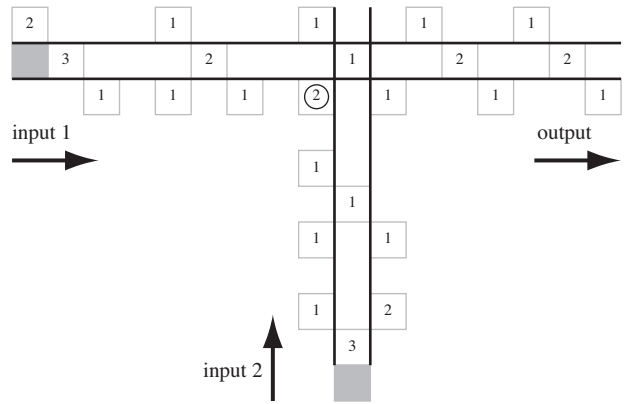}
    \caption{The logic AND gate for two inputs for $p=3$. The circled ``2''
      is put in order to get enough tokens for the horizontal and
      vertical inputs.}
    \label{fig:and-p=3}
  \end{figure}

\begin{figure}
  \centering
  \includegraphics[scale=.55]{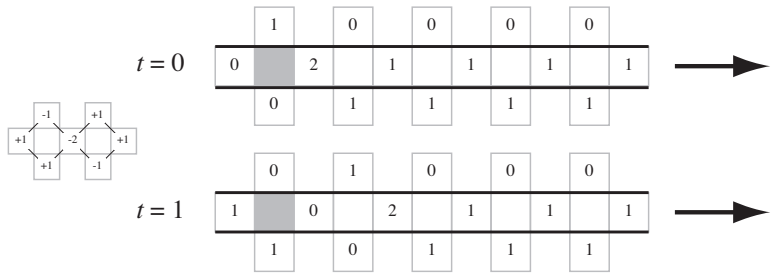}
  \caption{Information propagation in a wire for $p=2$ at times $t=0$
    and $t=1$ using the templates for Baks dynamics (recalled on the
    left of the figure).}
  \label{fig:wire}
\end{figure}
\begin{figure}
  \centering
  \includegraphics[scale=.55]{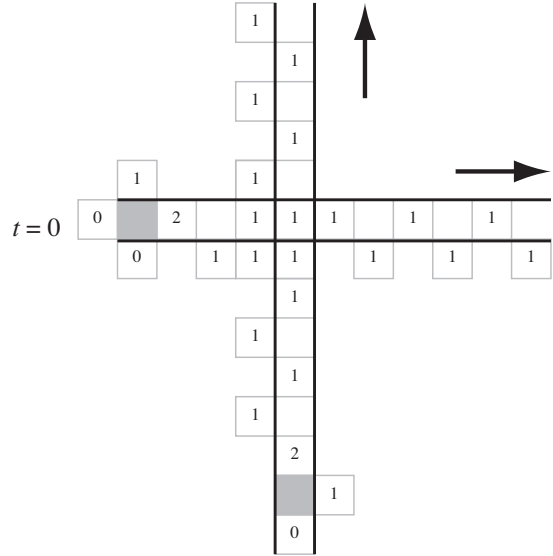}
  \caption{Crossing over two wires for $p=2$; arrows show the
    directions of propagation.}
  \label{fig:crossover}
\end{figure}
\begin{figure}
  \centering
  \includegraphics[scale=.55]{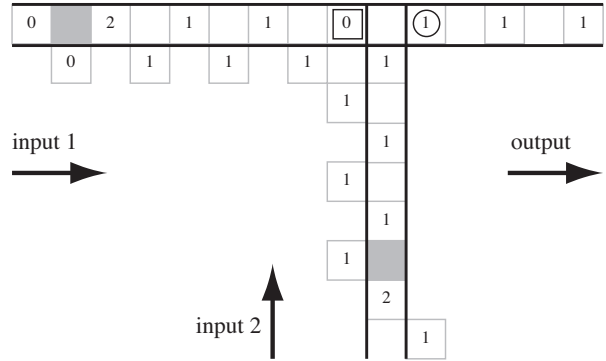}
  \caption{A logic AND gate with two inputs for $p=2$. The upcoming ``2''
    has to reach the horizontal ``2'' to change the value of the boxed
  ``0'' to ``1''. Then, the upcoming ``2'' can apply the vertical chenille
  template and changes the circled ``1'' into ``2''. In other words,
  the AND is computed by applying 3 horizontal chenilles and 4
  vertical ones.}
  \label{fig:andgate}
\end{figure}
\begin{figure}
  \centering
  \includegraphics[scale=.55]{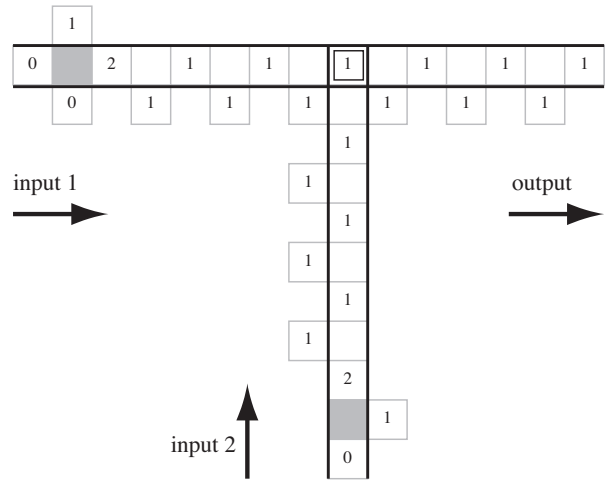}
  \caption{A logic OR gate with two inputs for $p=2$. The boxed cell
    indicates the OR gate point of computation.}
  \label{fig:orgate}
\end{figure}
\begin{figure}
  \centering
  \includegraphics[scale=.55]{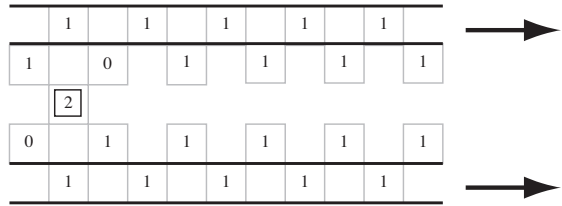}
  \caption{A signal multiplier for $p=2$. The signal starts on boxed
    site with value 2 and the first vertical chenille ruled by the
    Kadanoff's dynamics multiplies the signal on both horizontal wires.}
  \label{fig:sigmul}
\end{figure}

\begin{figure}
  \centering
  \includegraphics[scale=.45]{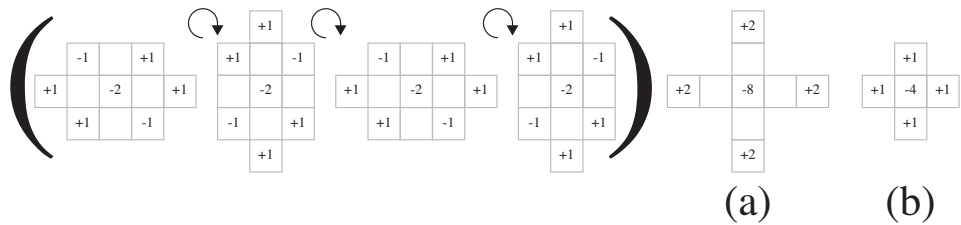}
  \caption{From Bak's to Kadanoff's operators. All the Kadanoff's
    operators between the brackets have been applied to get the
    pattern (a). The Bak's pattern (b) is obtained by eliminating the
    holes in (a) and by dividing the number of tokens by two.}
  \label{fig:BakRotate}
\end{figure}

\end{document}